\documentclass[11pt]{article}

\usepackage[margin=1in]{geometry}
\usepackage[utf8]{inputenc}
\usepackage[T1]{fontenc}
\usepackage[hyphens]{url}
\usepackage{graphicx}
\usepackage{booktabs}
\usepackage{amsmath}
\usepackage{amssymb}
\usepackage{amsthm}
\usepackage{algorithm}
\usepackage{algorithmic}
\usepackage{caption}
\usepackage[round,authoryear]{natbib}
\usepackage{hyperref}
\hypersetup{colorlinks=true,linkcolor=blue,citecolor=blue,urlcolor=blue}
\newtheorem{proposition}{Proposition}

\newtheorem{assumption}{Assumption}
\newtheorem{remark}{Remark}

\newcommand{\epsmax}{\varepsilon_{\max}}
\newcommand{\sigmahat}{\hat{\sigma}}
\newcommand{\Bt}{B_t}
\newcommand{\csrnr}{CS-RNR}

\title{Agents That Certify Their Own Exploits:\\
Confidence-Scheduled Restricted Responses\\
for Safe Opponent Exploitation}

\author{Boning Li\thanks{IIIS, Tsinghua University.}
\and Longbo Huang\thanks{IIIS, Tsinghua University. Corresponding author.}}

\date{\today}

\begin{document}

\maketitle

\begin{abstract}
An agent playing a Nash-equilibrium strategy in a two-player zero-sum
imperfect-information game secures the game value but forfeits the additional
value offered by a flawed opponent. Diffuse deviations pose a particular challenge: binary
release rules may gather too little evidence to act, while a full best
response to an incomplete opponent model can be highly exploitable. We
introduce \emph{budget-constrained confidence-scheduled restricted responses}
(\csrnr{}), the first opponent-exploitation method whose safety guarantee is a
certificate the agent computes on the strategy it actually deploys, so that
every exploit it commits to is one it has audited itself.
The method tracks pooled action frequencies with anytime-valid
confidence sequences and treats a frequency as exploitable only once its
interval separates from an equilibrium reference. The confirmed deviations
define a conservative opponent model, which a restricted-response solve turns
into candidate counter-strategies over a grid of pin levels. Before deployment,
each complete candidate is evaluated by a full-tree best response. The
resulting certificate is compared with a user-specified budget and committed
atomically with the strategy. Because this check is performed on the played
strategy, model quality determines the exploitation achieved while the
certificate controls reference-relative expected loss. In Leduc hold'em,
\csrnr{} obtains $6.2\times$ the steady-state gain of a money-verified binary
gate while keeping every deployed strategy within budget. A trajectory
mixture using the same estimator reaches $13.6\times$ the budget. Across
Leduc, Liar's Dice, and 5-rank Leduc, all 36{,}000 audited hands satisfy the
reported certificate tolerance.

\end{abstract}

\section{Introduction}

Two-player zero-sum imperfect-information games admit a robust default: a
Nash-equilibrium strategy guarantees the game value against any opponent.
This maximin property underlies the superhuman poker agents of the last
decade \citep{bowling2015heads,moravvcik2017deepstack,brown2018superhuman,
brown2019superhuman}.
Against a fixed non-equilibrium opponent, however, the same strategy can
forgo most of the available value \citep{southey2005bayes,hoehn2005effective}.
An exploiter can recover that value by
deviating from equilibrium, but the deviation also creates opportunities for
counter-exploitation \citep{mccracken2004safe,johanson2007computing,
ganzfried2015safe}. We study how to allocate a fixed exploitability budget
from the games observed online.

The problem is pronounced for \emph{diffuse} deviations, where behaviour
moves slightly away from equilibrium across many decision points rather than
failing at one conspicuous action. Such deviations arise naturally in human
play \citep{southey2005bayes,ge2022modeling} and imperfectly converged agents
\citep{zinkevich2007regret,tammelin2014solving}. On our Leduc suite, a money-verified
binary gate releases on none of the 12 diffuse opponents at the main horizon
and continues paying for probes while it waits. A continuous schedule that
immediately best-responds to a confidence-bound model does act, but reaches
$-0.66$ chips per hand against one diffuse opponent. Both failures stem from
partial identification: the gate demands enough evidence for a full release,
whereas the continuous schedule deploys an extreme response before the rest
of the opponent strategy is known.

We propose confidence-scheduled restricted responses (\csrnr{}). The method
uses anytime-valid confidence sequences
\citep{howard2021time,grunwald2024proposer} to
identify pooled action frequencies that separate from an equilibrium
reference. Confirmed excesses define a conservative opponent model. A
restricted-response solve
\citep{johanson2007computing,johanson2009data,ponsen2010mcrnr}
then produces a candidate strategy at pin level $p$, interpolating between
the equilibrium reference and a model-specific response. A full-tree best
response evaluates the original-game downside of the complete candidate.
Only candidates whose certificate is within the budget $\epsmax$ are
committed for play. Figure~\ref{fig:method-loop} summarizes this online loop.

\begin{figure}[t]
\centering
\includegraphics[width=0.56\textwidth]{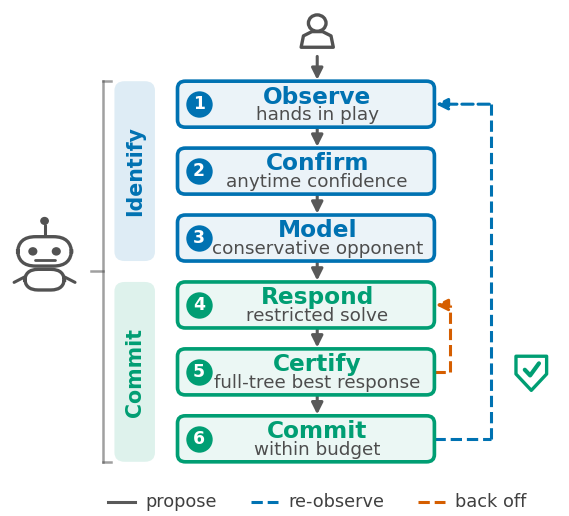}
\caption{The \csrnr{} deployment loop. Online observations update
anytime-valid confidence sequences and a conservative opponent model. A
restricted solve proposes a complete candidate at pin level $p$. The
original-game full-tree best response supplies its certificate $B$; only a
candidate satisfying $B\le\epsmax$ is committed with that certificate.
Failed candidates return to a lower pin level, and deployed play supplies the
next observations.}
\label{fig:method-loop}
\end{figure}

The restricted-response frontier explains why the pin level should be
scheduled from data. Against a strong diffuse deviator, $p{=}0.5$ captures
$46\%$ of the best-response gain at $6\%$ of its exploitability, while the
frontier changes substantially across opponents. Across 24 scripted Leduc
opponents, \csrnr{} earns $6.2\times$ the steady-state gain of the binary
gate and keeps every deployed certificate within budget. The same comparison
yields a $5.8\times$ gain ratio on Liar's Dice. A trajectory mixture using
the same estimator reaches $13.6\times$ the Leduc budget. Adversarial audits
covering 36{,}000 hands across Leduc, Liar's Dice, and 5-rank Leduc find no
certificate violation. The larger Leduc game exhibits confirmation
starvation at short horizons; longer horizons increase the number of
confirmed opponents while all deployed strategies remain within budget.

This relocates safety from the model to the strategy, and in doing so makes the
agent the auditor of its own exploits. Prior safe-exploitation guarantees are
earned by how the response is built
\citep{mccracken2004safe,johanson2009data,ganzfried2015safe,brown2017safe},
so they degrade exactly when the opponent
model is wrong: the agent inherits a promise made on its behalf by the
procedure that proposed its strategy. Our agent instead measures the
certificate on the complete candidate after it is built and before it is
played, and withholds any candidate that fails. That check is immune to model
error by construction, so a mistaken model can cost gain, never safety. The
distinction is measurable. Fixed-Mix uses the identical estimator and pin
level yet plays a strategy of exploitability $2.04$ against \csrnr{}'s $0.15$,
overshooting its budget by $14\times$, $4.4\times$, and $10\times$ across our
three games, because a trajectory mixture inherits the extreme exploitability
of the best response it mixes in.

Our contributions are:
\begin{itemize}
\item We develop \csrnr{}, the first online exploitation algorithm to combine
  anytime-valid deviation confirmation, conservative opponent modeling,
  restricted-response solving, and atomic deployment under a computed budget
  certificate.
\item We establish a runtime invariant that bounds the conditional expected
  loss of every deployed strategy. Unlike safety analyses that bound loss
  through the construction of the response, our invariant is evaluated on the
  deployed strategy itself, and the analysis separates this strategy-level
  certificate from the statistical conditions governing deviation detection.
\item We evaluate six deployment rules across three games, including adaptive
  audits, budget sweeps, and horizon scaling. The results identify safe gain
  and certificate separation from fixed mixtures, and establish confirmation
  starvation as a distinct, data-dependent limit of safe exploitation.
\end{itemize}

\section{Related Work}

\paragraph{Safe opponent exploitation.}
\citet{mccracken2004safe} bound the loss an agent-modelling strategy may incur
relative to the equilibrium guarantee, and \citet{hoehn2005effective} measure
the exploration cost of exploiting an unknown opponent over a short match.
\citet{ganzfried2015safe} study repeated zero-sum play with strategies that
trade exploitation against an explicit safety allowance informed by an
opponent model and match history, and \citet{bernasconi2021exploiting} exploit
opponents in sequential games subject to utility constraints. Restricted Nash
responses
\citep{johanson2007computing} and data biased responses
\citep{johanson2009data} compute robust counter-strategies by pinning the
opponent to a model with probability $p$, and sampling makes the pinned solve
practical on larger trees \citep{ponsen2010mcrnr}. We use this pinned-opponent
solve as the response primitive, evaluate each candidate by an original-game
best response, and schedule the measured downside. Safe subgame solving bounds the loss a
re-solve may add to a blueprint \citep{brown2017safe}, and remains sound under
a depth limit \citep{brown2018depth}; subgame methods
\citep{liu2022safe,ge2023efficient,ge2024safe} carry this control into
exploitative re-solving, \citet{liu2023opponent} restrict the opponent during
online search, \citet{wang2025horse} organize exploitation hierarchically, and
\citet{milec2025adapting} study counter-strategies under depth limits; at poker
scale the cost of each such solve is governed by abstraction, pruning
\citep{brown2015regret,brown2017dynamic}, and parallel solving
\citep{li2024rl,li2025efficient,li2026effective,li2026real}. We work
in solver-tractable games, where the certificate is an exact full-tree quantity,
and schedule a whole-game policy between repeated hands.

\paragraph{Learning-based opponent modeling.}
Offline population training can produce exploiters that infer an opponent
online, including greedy-when-sure policies with Bayesian-style inference
\citep{fu2022greedy}, rationality models fitted to unknown agents
\citep{ge2022modeling}, in-context identification from interaction history
\citep{jing2023towards,jing2024opponent,jing2025open}, and poker exploiters
trained to punish suboptimal play
\citep{murgoci2026alphaexploitem,caen2026stratformer}; see
\citet{nashed2022survey} for a survey. This line emphasizes adaptation and
empirical return.
Across this line the agent acts on a learned model of its opponent and its
exposure is whatever that model implies, which is the coupling our certificate
removes.
We instead use online frequency evidence to propose candidates and an
original-game certificate to govern their deployment, without offline
training.

\paragraph{Sequential testing in games.}
Anytime-valid confidence sequences
\citep{howard2021time,grunwald2024proposer} support continuous monitoring
with time-uniform error control. Recent work applies sequential
tests to strategic-deviation monitoring in multi-agent systems
\citep{gauthier2026betting}, punishment triggers in repeated games
\citep{capitaine2026test}, and testing whether an agent follows a given
mixed strategy \citep{ganzfried2025nonparametric}. Consistent opponent
modeling of static opponents has also been studied
\citep{DBLP:journals/corr/abs-2508-17671}. We use sequential evidence
to construct a conservative model, then solve and certify the resulting
strategy before deployment.

\section{Preliminaries}

\paragraph{Setting.}
We consider a two-player zero-sum extensive-form game with imperfect
information and perfect recall. A behavioural strategy $\sigma_i$ maps
each information set $I$ of player $i$ to a distribution over legal
actions; $\sigma=(\sigma_0,\sigma_1)$ denotes a profile and
$u_i(\sigma)$ player $i$'s expected value. The game value to player $i$
of a Nash equilibrium $\sigma^*$ is $v_i^*$, and the
\emph{exploitability} of a strategy $\sigma_i$ is
\[
\mathrm{expl}(\sigma_i) \;=\;
\max_{\sigma_{-i}} u_{-i}(\sigma_i,\sigma_{-i}) - v^*_{-i}
\;\ge\; 0,
\]
the amount a best-responding adversary wins beyond its game value. A Nash
strategy has zero exploitability; playing it forfeits, against a fixed
non-equilibrium opponent $\sigma_{\mathrm{opp}}$, up to the
\emph{best-response gain}
$g^*(\sigma_{\mathrm{opp}})=\max_{\sigma_i}
u_i(\sigma_i,\sigma_{\mathrm{opp}})-v^*_i$.

\emph{Repeated play.} The agent (hero) plays $T$ hands against an unknown
opponent, observing the public actions of each hand. Hero may change its
strategy between hands. We report per-hand gain over the Nash floor,
$u_i(\sigma_i^{(t)},\sigma_{\mathrm{opp}})-v_i^*$, computed exactly for
the strategy played at hand $t$. Probe and exploration costs are therefore
included in every curve.

\paragraph{Restricted responses.}
A restricted Nash response \citep[RNR]{johanson2007computing} against a
model $\sigmahat$ is the hero side of an equilibrium of the modified game
in which, with probability $p$, the opponent is forced to play $\sigmahat$,
and with probability $1-p$ plays freely. As $p$ ranges over $[0,1]$, the
hero strategy interpolates from Nash to $\mathrm{BR}(\sigmahat)$. Data
biased responses \citep[DBR]{johanson2009data} apply the pin per information
set rather than per trajectory. We use the behavioural DBR form: the
restricted seat plays $p\,\sigmahat_I + (1-p)\,\sigma^{\mathrm{rm}}_I$ at
each information set $I$, where $\sigma^{\mathrm{rm}}$ is the free
regret-matching strategy. This form accepts a per-infoset vector of pin
probabilities, the granularity required by our scheduler. We use
``restricted response'' and the RNR initials for the family of
pinned-opponent counter-strategies. Our implementation modifies two lines of
a CFR$^+$ solver \citep{zinkevich2007regret,tammelin2014solving}: the
restricted seat uses the $p$ mixture above, and its instantaneous regret is
measured against the free component's value. Feasible deviations of the free
agent scale the regret vector by the constant $(1-p)$, which regret matching
ignores. Any regret minimizer serves as the primitive: predictive updates
\citep{farina2021faster,farina2023regret,meng2025efficient}, dynamic
discounting and scheduled or learned hyperparameters
\citep{brown2019solving,xu2024dynamic,sychrovsky2024learning,meng2026faster},
deep realizations \citep{brown2019deep,brown2020combining,xu2026deep}, and
variance-reduced chance sampling
\citep{li2026correlatedchancesamplingmonte} all return a strategy our
certificate then evaluates unchanged.

\section{Budget-Constrained Confidence-Scheduled Restricted Responses}
\label{sec:method}

\csrnr{} has three components. Detection identifies supported deviations,
response converts them into a candidate strategy, and certification controls
deployment. Certification evaluates the candidate itself, so detection and
modeling determine the value captured while the deployment bound remains
strategy-based (Figure~\ref{fig:method-loop}). Algorithm~\ref{alg:csrnr} is the
episode loop, and Algorithms~\ref{alg:detect} and~\ref{alg:sched} expand the two
steps it calls.

\begin{algorithm}[t]
\caption{\csrnr{}: one episode of $T$ hands against one opponent. The state
carried across hands is the played strategy $\sigma_H$, its certificate $B$,
the pin index $k$, and the pooled counts $n$.}
\label{alg:csrnr}
\begin{algorithmic}[1]
\REQUIRE equilibrium $\sigma^*$; reference value
$\widetilde v_{\mathrm{opp}}$; budget $\epsmax$; pin grid
$0=p_0<\cdots<p_K$; checkpoints $\mathcal{T}$; climb cap $c=2$
\STATE $\sigma_H\gets\sigma^*$;\ $k\gets 0$;\
$B\gets\widetilde B(\sigma^*)$;\ $n\gets\mathbf{0}$
\FOR{$t=1$ \TO $T$}
  \STATE play hand $t$ with $\sigma_H$;\ $\Bt\gets B$
  \STATE add the opponent's actions in hand $t$ to the pooled counts $n$
  \IF{$t\in\mathcal{T}$}
    \STATE $\mathcal{E}\gets\textsc{Detect}(n)$
      \COMMENT{Algorithm~\ref{alg:detect}}
    \IF{$\mathcal{E}=\emptyset$}
      \STATE $k\gets 0$;\ $\sigma_H\gets\sigma^*$;\
      $B\gets\widetilde B(\sigma^*)$
    \ELSE
      \STATE $\sigmahat\gets\textsc{Model}(\mathcal{E})$
        \COMMENT{Algorithm~\ref{alg:detect}}
      \STATE $(k,\sigma_H,B)\gets\textsc{Schedule}(\sigmahat,k)$
        \COMMENT{Algorithm~\ref{alg:sched}}
    \ENDIF
  \ENDIF
\ENDFOR
\end{algorithmic}
\end{algorithm}

\subsection{Detection: Anytime-Valid Deviation Excesses}
\label{sec:method-detect}

Observations are aggregated into \emph{pools}: information sets sharing a
round, a card or hand-strength category, and whether the seat is facing a
bet. In Leduc this gives 12 pools, preflop by private card and postflop by
hand-strength category, each split by whether the seat faces a raise. Each
pool carries an equilibrium reference frequency $f^*$ computed by a
reach-weighted pass of the same estimator over Nash-vs-Nash play, so the
monitored statistic and its null anchor share one construction. After every
hand, the opponent's observed actions update pooled counts. At
scheduling checkpoints the agent scans every (pool, action) cell with a
time-uniform confidence sequence for the pool frequency. We use stitched
sub-Gaussian or empirical-Bernstein boundaries
\citep{howard2021time} at per-cell level
$\alpha_{\mathrm{cell}}{=}\alpha/m$, Bonferroni-split across the $m{=}30$
scanned cells of 12 pools. The probability that any cell is ever falsely
flagged during the episode is therefore at most $\alpha{=}0.05$. A cell is
flagged when its interval lies beyond the pool's equilibrium frequency $f^*$
by more than margin $\delta{=}0.10$:
\[
[\hat f-w_n,\hat f+w_n]\cap[f^*-\delta,f^*+\delta]=\emptyset,
\]
where $w_n$ is the boundary half-width after $n$ observations of the pool at
level $\alpha_{\mathrm{cell}}$. Time uniformity permits evaluation after each
hand under the predictable sampling conditions stated in
Section~\ref{sec:theory}.

For a flagged cell, the \emph{signed conservative excess} is
$e=\mathrm{endpoint}-f^*$, where the endpoint is the interval edge nearest
the equilibrium frequency. This is the smallest deviation retained by the
interval. At every information set $I$ in a flagged pool $P(I)$, the model is
\[
\sigmahat_I(a)\propto
\mathrm{clip}\!\left(\sigma^*_I(a)+e_{P(I),a},0,1\right),
\]
with unflagged actions renormalized proportionally to the equilibrium
reference; elsewhere the model equals the reference. The additive key-level
shift preserves the confirmed direction across member information sets with
heterogeneous equilibrium frequencies. Substituting a common absolute pool
endpoint can reverse that direction when a member's local reference frequency
already exceeds the endpoint. Algorithm~\ref{alg:detect} states both steps.

\begin{algorithm}[t]
\caption{\textsc{Detect} and \textsc{Model}. The scan is two-sided and
$\alpha$ is split across all scanned cells, so the episode-wide false-flag
probability is at most $\alpha$. Boundary constants are in the appendix.}
\label{alg:detect}
\begin{algorithmic}[1]
\REQUIRE cells $\mathcal{C}=\{(P,a)\}$; references $f^*$; level $\alpha$;
margin $\delta$
\STATE \textbf{\textsc{Detect}($n$):}
\STATE $\alpha_{\mathrm{cell}}\gets\alpha/|\mathcal{C}|$;\
$\mathcal{E}\gets\emptyset$
\FORALL{$(P,a)\in\mathcal{C}$ with observed pool mass $n_P>0$}
  \STATE $\hat f\gets n_{P,a}/n_P$;\
  $w\gets w_{n_P}(\alpha_{\mathrm{cell}})$
  \IF{$\hat f+w<f^*_{P,a}-\delta$}
    \STATE $e_{P,a}\gets(\hat f+w)-f^*_{P,a}$
      \COMMENT{$<0$: plays $a$ less than equilibrium}
  \ELSIF{$\hat f-w>f^*_{P,a}+\delta$}
    \STATE $e_{P,a}\gets(\hat f-w)-f^*_{P,a}$
      \COMMENT{$>0$: plays $a$ more}
  \ENDIF
\ENDFOR
\RETURN $\mathcal{E}$, the flagged cells and their excesses
\item[]
\STATE \textbf{\textsc{Model}($\mathcal{E}$):}
\STATE $\sigmahat\gets\sigma^*$
\FORALL{$I$ with $P(I)$ flagged}
  \STATE $F\gets\{a:(P(I),a)\in\mathcal{E}\}$
  \STATE $\sigmahat_I(a)\gets
  \mathrm{clip}(\sigma^*_I(a)+e_{P(I),a},0,1)$ for $a\in F$
  \STATE $r\gets\max(1-\sum_{a\in F}\sigmahat_I(a),0)$;\
  $Z\gets\sum_{a\notin F}\sigma^*_I(a)$
  \STATE $\sigmahat_I(a)\gets r\,\sigma^*_I(a)/Z$ for $a\notin F$
    \COMMENT{uniform $r/|{\cdot}|$ if $Z=0$}
  \STATE normalize $\sigmahat_I$
\ENDFOR
\RETURN $\sigmahat$
\end{algorithmic}
\end{algorithm}

\begin{algorithm}[t]
\caption{\textsc{Schedule}: bounded local search over pin levels. Every
level reached is certified individually, and the returned pair is the one
committed, so no level is played before its own certificate clears
$\epsmax$.}
\label{alg:sched}
\begin{algorithmic}[1]
\REQUIRE model $\sigmahat$; current index $k$; grid $p_0,\dots,p_K$;
budget $\epsmax$; climb cap $c=2$; solve iterations $N=400$
\STATE \textbf{\textsc{Eval}($p$):} \textbf{if} $p=0$ \textbf{return}
$(\sigma^*,\widetilde B(\sigma^*))$; solve the restricted game at pin $p$
for $N$ CFR$^+$ iterations to get $\sigma_H(p)$;
\textbf{return} $(\sigma_H(p),\widetilde B(\sigma_H(p)))$
\item[]
\STATE $(\sigma,B)\gets\textsc{Eval}(p_k)$
  \COMMENT{re-certify the level in force against the updated model}
\WHILE{$B>\epsmax$ \AND $k>0$}
  \STATE $k\gets k-1$;\ $(\sigma,B)\gets\textsc{Eval}(p_k)$
\ENDWHILE
\STATE $j\gets 0$
\WHILE{$B\le\epsmax$ \AND $k<K$ \AND $j<c$}
  \STATE $(\sigma',B')\gets\textsc{Eval}(p_{k+1})$;\ $j\gets j+1$
  \IF{$B'>\epsmax$}
    \STATE \textbf{break}
      \COMMENT{candidate discarded; nothing is committed}
  \ENDIF
  \STATE $k\gets k+1$;\ $(\sigma,B)\gets(\sigma',B')$
\ENDWHILE
\RETURN $(k,\sigma,B)$
\end{algorithmic}
\end{algorithm}

\subsection{Response: Restricted Solve at Pin Level $p$}

Given $\sigmahat$, the agent solves the restricted game at pin level $p$ to
obtain $\sigma_H(p)$. The opponent seat follows $\sigmahat$ with weight $p$
and remains free with weight $1-p$, producing a hedged behavioural strategy
rather than an unrestricted best response to the model. The offline frontier
in Figure~\ref{fig:frontier} shows that substantial gain can occur at low
exploitability and that the useful pin level varies by opponent.

\subsection{Certification: The Budget Rule}
\label{sec:method-cert}

For a complete candidate strategy $\sigma_H$, the agent computes the
reference-relative certificate
\[
\widetilde B(\sigma_H)
=\max_{\sigma_{\mathrm{opp}}}
u_{\mathrm{opp}}(\sigma_H,\sigma_{\mathrm{opp}})
-\widetilde v_{\mathrm{opp}},
\]
where the maximum is an original-game full-tree best-response evaluation and
$\widetilde v$ is the fixed CFR-derived reference value. The computation is
exact for the played strategy in the solver-tractable games studied here;
Section~\ref{sec:theory} gives the correction for finite-reference error.

The scheduler (Algorithm~\ref{alg:sched}) searches the grid
$p\in\{0,0.1,0.2,0.3,0.5,0.7,0.9\}$ at checkpoints that thin out over the
episode, hands $\{40,80,160,320,480,640\}$ for $T{=}800$. At each checkpoint it
re-certifies
the current level against the updated model, steps down while the budget is
violated, and tests at most two higher levels. Each candidate reached by this
bounded local search is certified individually by one exact best-response pass
over a restricted solve of 400 CFR$^+$ iterations, so the rule does not require
monotonic exploitability in $p$ or a globally maximal feasible grid point.
Between checkpoints, hero plays the accepted $\sigma_H$ unchanged. The
certificate in force at hand $t$, denoted $B_t$, is therefore computed for
the strategy played at that hand.

The deployed state is an atomic pair $(\sigma_H,B)$. A complete candidate is
held in a temporary buffer and committed only after satisfying
$B\le\epsmax$. Failed candidates are discarded, and execution applies no
mixing, truncation, or missing-information-set fallback after certification.
Statistical error can select an unprofitable candidate, but its certificate
is computed without $\sigmahat$. The parameter $\epsmax$ controls the
continuum between Nash play and the unrestricted schedule.

\section{Guarantees}
\label{sec:theory}

The analysis separates the deployment invariant from the statistical
conditions used to identify profitable deviations.

\begin{proposition}[Runtime certificate invariant]
\label{prop:cert}
Suppose (a) at every hand $t$ the runtime maintains an atomic deployed pair
$(\sigma_H^{(t)},B_t)$, where $B_t$ is computed by a full-tree best response
for the complete behavioural strategy $\sigma_H^{(t)}$, and the executor
plays that strategy without post-certification modification; and (b) the
schedule admits a candidate only if $B_t\le\epsmax$. If $B_t$ is defined
relative to the exact game value, then for every opponent strategy
$\sigma^{(t)}_{\mathrm{opp}}$ at hand $t$, including one selected adaptively
with knowledge of $\sigma_H^{(t)}$ and the agent state, the conditional
expected value satisfies
\[
u_{\mathrm{hero}}(\sigma_H^{(t)},\sigma^{(t)}_{\mathrm{opp}})
\ge v^*_{\mathrm{hero}}-B_t\ge v^*_{\mathrm{hero}}-\epsmax.
\]
Consequently, expected mean per-hand gain over any horizon is at least
$-\epsmax$. Realized winnings also contain the game's per-hand variance.
\end{proposition}

\begin{proof}
By definition,
$B_t=\max_{\sigma_{\mathrm{opp}}}
u_{\mathrm{opp}}(\sigma_H^{(t)},
\sigma_{\mathrm{opp}})-v^*_{\mathrm{opp}}$. For any
$\sigma^{(t)}_{\mathrm{opp}}$, zero-sum payoffs give
$u_{\mathrm{hero}}=-u_{\mathrm{opp}}\ge
-(v^*_{\mathrm{opp}}+B_t)=v^*_{\mathrm{hero}}-B_t$.
Condition (b) gives the horizon statement even when $B_t$ is random and
history-dependent.
\end{proof}

\paragraph{Finite-reference correction.}
The implementation uses a finite-iteration reference value
$\widetilde v_{\mathrm{opp}}$. Let
$|\widetilde v_{\mathrm{opp}}-v^*_{\mathrm{opp}}|\le\eta_v$ and define
$\widetilde B_t=\max_{\sigma_{\mathrm{opp}}}
u_{\mathrm{opp}}
(\sigma_H^{(t)},\sigma_{\mathrm{opp}})-\widetilde v_{\mathrm{opp}}$.
Then
$u_{\mathrm{hero}}\ge v^*_{\mathrm{hero}}-\widetilde B_t-\eta_v$.
A runtime threshold $\widetilde B_t\le\epsmax$ therefore certifies true-game
loss at most $\epsmax+\eta_v$. The tree-walk best response is exact for the
fixed played strategy, while $\eta_v$ records the separate reference error.
The every-hand best responder in the stress test realizes this certificate
boundary within solver tolerance.

\begin{assumption}[Estimator idealization]
\label{ass:estimator}
Proposition~\ref{prop:cs} considers per-(pool, action) observation processes
whose increments have conditional mean $f$ given the past, where $f$ is the
reach-weighted pool frequency under a fixed hero strategy. Hidden-card
soft-attribution weights and the within-pool reach mixture are treated as
fixed. Under the null, the implementation computes $f^*$ with the same
attribution and reach weighting. Under deviation, margin $\delta$ also absorbs
attribution error and reach-mixture drift induced by replanning.
Proposition~\ref{prop:cert} holds independently of this assumption.
\end{assumption}

\begin{proposition}[Detection validity]
\label{prop:cs}
Under Assumption~\ref{ass:estimator}, against any fixed opponent
$\sigma_{\mathrm{opp}}$, with probability at least $1-\alpha$ every
confidence sequence simultaneously covers its monitored pool frequency at
all times. On this event: (i) no cell whose frequency lies within $\delta$ of
the equilibrium reference is flagged, so an exact-Nash opponent keeps
\csrnr{} at $p=0$; and (ii) every flagged excess understates the monitored
pool-level deviation. Under within-pool homogeneity, the per-infoset model
inherits its direction and a lower bound on its magnitude.
\end{proposition}

\begin{proof}[Proof sketch]
Each cell uses a time-uniform boundary at level $\alpha/m$
\citep{howard2021time}. A union bound over the $m$ cells gives
simultaneous coverage $1-\alpha$. On this event, an interval containing the
monitored frequency can lie entirely beyond $f^*\pm\delta$ only when the
frequency does. The interval edge nearest $f^*$ lies between $f^*$ and that
frequency, giving the conservative excess.
\end{proof}

At individual decision points, the schedule is predictable because each
checkpoint uses earlier data, and hero's strategy changes which opponent
points are observed. Time-uniform boundaries permit such optional
observation. Replanning also changes the reach mixture within a pool, as
captured by Assumption~\ref{ass:estimator} and the margin. The confidence
sequences describe fixed opponents; adaptive-opponent deployment remains
covered by Proposition~\ref{prop:cert}.

\begin{remark}[Certificate gap for trajectory mixtures]
A trajectory-level mixture plays $\mathrm{BR}(\sigmahat)$ with probability
$p$ and the reference strategy otherwise. Its exploitability is at most
$p\,\mathrm{expl}(\mathrm{BR}(\sigmahat))$ by linearity of opponent value in
hero's mixture. The factor $\mathrm{expl}(\mathrm{BR}(\sigmahat))$ can remain
near the game's maximum because a best response to a nearly Nash model is an
extreme tie-breaking strategy. In our main comparison, Fixed-Mix reaches
exploitability $2.04$ at $p=0.5$. The behavioural restricted solve at the
same level and with the same $\sigmahat$ reaches $0.15$. \csrnr{} schedules
the latter measured quantity.
\end{remark}

\section{Experiments}
\label{sec:experiments}

Every reported gain is a full tree traversal of the profile in force at that
hand, so probe costs, cold-start losses, and scheduling mistakes are visible
without evaluation noise; sampling randomness enters only through the
observations the online learner sees. The main testbed is Leduc hold'em
\citep{southey2005bayes} with unit ante, a standard 288-information-set
benchmark whose small tree permits full-tree evaluation of every reported
strategy (units chips/hand; maximal exploitability ${\approx}6.1$, and a strong
deviator offers a best response from $0.7$ to $3.5$). Observations aggregate into
12 pools; game rules, the pool taxonomy, and the detector configuration are in
Appendix~\ref{app:repro}. Section~\ref{sec:generality} repeats the study on
Liar's Dice and 5-rank Leduc.

\paragraph{Opponents.}
Two 12-opponent suites. \emph{Concentrated}: equilibrium play except one
or two isolated, pure-action blunders (e.g.\ folds the nuts facing a
bet), the classical target of safe exploitation. \emph{Diffuse}: log-odds
perturbations of the CFR-derived equilibrium reference across all decision
points, six behavioural archetypes at two strengths, worth from $0.20$ to
$1.84$ to a best response yet with no single blunder to find. Stress tests
additionally use adaptive opponents (below).

\paragraph{Arms.}
(i)~\textsc{Nash}: the equilibrium, gain $0$ by construction;
(ii)~\textsc{Oracle}: static exact $\mathrm{BR}(\sigma_{\mathrm{opp}})$,
the exploitation ceiling;
(iii)~\textsc{Binary gate}: play Nash, probe a candidate leak at rate
$0.1$, and release a full best response to the conservative model once
realized probe winnings confirm a money edge;
(iv)~\textsc{Fixed-Mix}: play $\mathrm{BR}(\sigmahat)$ with probability
$p{=}0.5$ and Nash otherwise. This trajectory mixture differs from the
pinned game of \citet{johanson2007computing};
(v)~\textsc{Fixed-DBR}: behavioural restricted solve at constant
$p{=}0.5$;
(vi)~\csrnr{} with budget $\epsmax{=}0.15$. Arms (iv), (v), and (vi)
share the estimator, checkpoints, and solver, isolating the schedule.
We use $T{=}800$ hands and $10$ seeds per opponent. Comparisons report mean
steady-state gain, losing opponents, and the largest exploitability reached.

\subsection{The Offline Frontier}

\begin{figure}[t]
\centering
\includegraphics[width=0.62\textwidth]{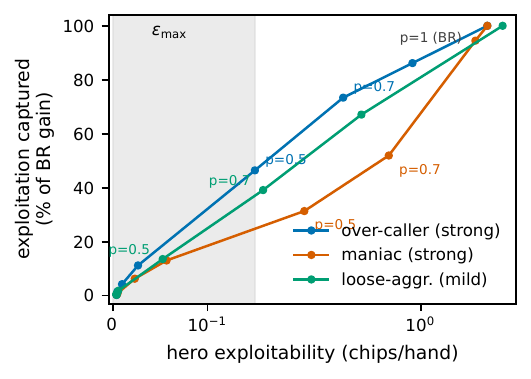}
\caption{Restricted-response frontier with an oracle model
($\sigmahat=\sigma_{\mathrm{opp}}$); $x$ symlog-scaled, shaded band:
$B \le \epsmax = 0.15$. Curves are strongly concave (the strong
over-caller: $46\%$ of the BR gain at $6\%$ of its exploitability) and
opponent-dependent (the mild deviator offers almost nothing inside the
band). Full BR against the \emph{mildest} opponent is the most
exploitable strategy shown ($3.01$).}
\label{fig:frontier}
\end{figure}

Figure~\ref{fig:frontier} maps what restricted responses buy when the
model is exact: exploitation concentrates at small $p$; every curve
stays under exploitability $0.06$ for $p \le 0.3$; and the frontier's
shape varies with the opponent, so the profitable operating point must
be found from data, which motivates the schedule.

\subsection{Main Comparison}

\begin{table}[t]
\centering
\small
\begin{tabular}{lrrrr}
\toprule
Arm & Gain & Expl. & Lose & Worst cert.\\
\midrule
\textsc{Nash} & $0$ & 0/24 & 0 & n/a\\
\textsc{Oracle} (ceiling) & $+1.169$ & 24/24 & 0 & $6.11$\\
\textsc{Binary gate} & $+0.034$ & 11/24 & \textbf{3} & uncert.\\
\textsc{Fixed-Mix} & $+0.227$ & 13/24 & 0 & \textbf{2.04}\\
\textsc{Fixed-DBR} & $+0.193$ & 13/24 & 0 & $0.143^{\dagger}$\\
\csrnr{} (ours) & $+0.209$ & 13/24 & 0 & $\mathbf{0.150} \le \epsmax$\\
\bottomrule
\end{tabular}
\caption{Main comparison, 24 opponents $\times$ 10 seeds, $T{=}800$,
$\epsmax{=}0.15$. Gain: mean exact steady-state gain (final fifth of
hands). Expl.: opponents exploited ($95\%$ CI above zero). Lose:
opponents with negative mean cumulative gain. Worst cert.: largest
exploitability of any strategy the arm ever played.
$\dagger$observed in this sweep but not enforced by the fixed pin.}
\label{tab:main}
\end{table}

Table~\ref{tab:main} gives the headline. \csrnr{} earns $6.2\times$ the
binary gate's steady-state gain (paired over $24$ opponents: $t{=}3.8$,
$p{<}0.002$; Wilcoxon $p{<}0.002$), has no negative mean cumulative
matchup at $T{=}800$, and its worst certificate over all $240$ episodes is
$0.1496$, so the budget binds without being crossed. The gate, despite money
verification, \emph{loses} against three opponents (worst $-0.068$) because it
never resolves: candidate leaks stay in probation forever and the
$\varepsilon$-probe hands spent there bleed money indefinitely (release
fraction $0$ on all three). Fixed-Mix matches
\csrnr{}'s gain but its worst played strategy has exploitability $2.04$,
$13.6\times$ the budget: identical estimator, no certificate, no safety.
Fixed-DBR stays low only because $p{=}0.5$ happens to be cheap against
these models, and its constant level costs it where the schedule climbs
higher ($+0.142$ vs.\ \csrnr{}'s $+0.158$ on the tag-flatter).
\begin{figure}[t]
\centering
\includegraphics[width=0.66\textwidth]{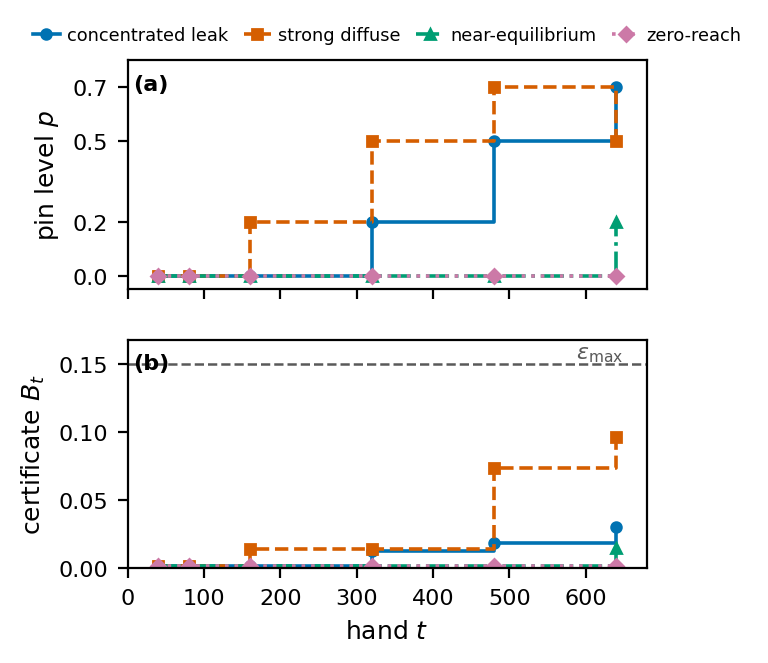}
\caption{Per-opponent schedule adaptation (median over $10$ seeds).
(a) Pin level $p$ against hand $t$: the schedule climbs on a strong
concentrated leak (to $0.7$) and a diffuse over-caller (to $0.7$, then a
checkpoint re-certification steps it back to $0.5$), reaches only $0.2$ on a
near-equilibrium deviator, and stays at $0$ where nothing is confirmable.
(b) The certificate $B_t$ of every deployed level stays under
$\epsmax{=}0.15$.}
\label{fig:schedule}
\end{figure}

Figure~\ref{fig:schedule} traces this per opponent: the level the schedule
settles on tracks how much the data confirm, and every certificate along the way
stays within budget.

\subsection{Certificate Stress Test}

\begin{figure}[t]
\centering
\includegraphics[width=0.62\textwidth]{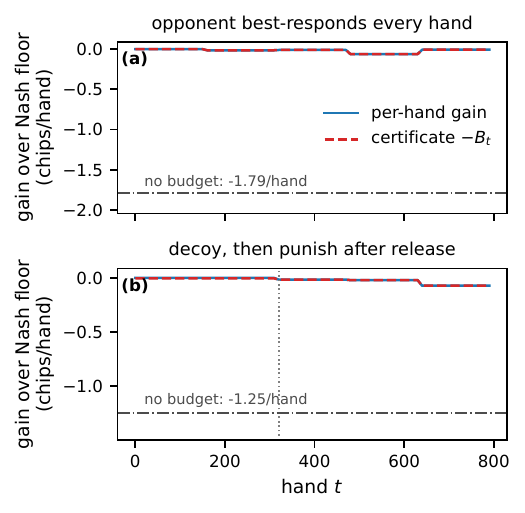}
\caption{Adversarial audit ($T{=}800$, $\epsmax{=}0.15$; hardest of 5
seeds). Solid: per-hand expected gain; dashed: $-\Bt$. (a) Against the
every-hand best responder the curves coincide. (b) Bait-and-punish
costs $-0.009$/hand. Dash-dotted: measured mean loss with the same
detector but no budget ($-1.79$ and $-1.25$ per hand).}
\label{fig:stress}
\end{figure}

\begin{figure}[t]
\centering
\includegraphics[width=0.60\textwidth]{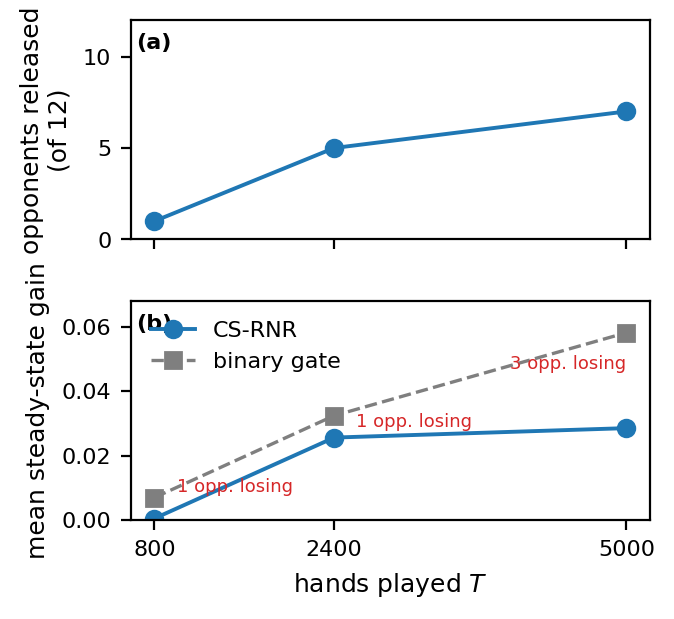}
\caption{Diffuse suite vs.\ horizon under empirical-Bernstein boundaries
($\delta{=}0.05$; scaled checkpoints, unchanged gate). (a) Opponents
released by \csrnr{}. (b) Mean steady-state gain. At $T{=}5000$, one
unbounded full-BR release raises the gate mean. \csrnr{} keeps every
certificate ${\le}0.148$; two marginal releases lose $-0.0022$ and $-0.0018$
per hand cumulatively, within budget.}
\label{fig:scaling}
\end{figure}

Proposition~\ref{prop:cert} applies to opponents outside the statistical model.
We evaluate three adaptive opponents with $5$ seeds each and audit
$\mathrm{gain}_t\ge-\Bt$ up to solver tolerance
(Appendix~\ref{app:repro}), alongside the same
system without the budget ($\epsmax=\infty$, releasing to $p=1$).
\emph{Every-hand BR} recomputes an exact best response whenever hero replans.
The audit finds no violations and expected loss matches the envelope: the worst
hand is $-0.0633$ at certificate $0.0633$ (Fig.~\ref{fig:stress}(a)), within the
$5\times10^{-3}$ solver tolerance. The unbudgeted system loses $-1.79$/hand,
$128\times$ larger, at an ignored certificate of $4.08$.
\emph{Decoy-then-punish} first presents a strong diffuse deviation and then
best-responds to hero after release. The detector raises $p$ before the
switch, and the certificate caps overall loss at $-0.009$/hand, while the
unbudgeted system loses $-1.25$/hand with individual hands reaching $-3.24$.
\emph{Nash-then-deviate} plays equilibrium for $400$ hands before a strong
deviation; detection is delayed, hero remains at Nash, and the audit is
unchanged. Across the three modes, all $12{,}000$ hands pass and the largest
certificate is $0.069$.

\subsection{The Price of Confirmation}

Diffuse deviations are statistically expensive to confirm: a pool-frequency
shift from $0.1$ to $0.28$ needs on the order of $10^3$ observations to
separate from equilibrium, while 800 hands supply only tens per pool. Under
the horizon sweep of Figure~\ref{fig:scaling}, \csrnr{} releases 1, 5, and 7
of 12 opponents as $T$ grows from 800 to 2400 to 5000 and keeps every deployed
certificate within budget, whereas the binary gate has from 1 to 3 losing
opponents across these horizons. Per-opponent gains and the boundary
comparison are in Appendix~\ref{app:repro}.

\subsection{Ablations: the Budget Dial vs.\ Fixed Pins}

\begin{table}[t]
\centering
\small
\begin{tabular}{lrr@{\hspace{10pt}}lrr}
\toprule
\csrnr{} & Gain & Cert. & \textsc{F-DBR} & Gain & Cert.\\
\midrule
$\epsmax{=}0.05$ & $+0.121$ & $0.048$ & $p{=}0.2$ & $+0.108$ & $0.026$\\
$\epsmax{=}0.15$ & $+0.160$ & $0.150$ & $p{=}0.3$ & $+0.122$ & $0.062$\\
$\epsmax{=}0.30$ & $+0.234$ & $0.254$ & $p{=}0.5$ & $+0.155$ & $0.143$\\
 & & & $p{=}0.7$ & $+0.235$ & $\mathbf{0.248}$\\
\bottomrule
\end{tabular}
\caption{Budget and pin sweeps on a 6-opponent subset spanning both
suites ($10$ seeds, $T{=}800$). Gain: mean steady-state gain. Cert.:
worst played exploitability. Every \csrnr{} row keeps its certificate
within its own budget; every fixed pin either under-exploits or, at
$p{=}0.7$, plays a strategy that would violate the $\epsmax{=}0.15$
budget by $65\%$ without a prospective enforcement mechanism.}
\label{tab:ablation}
\end{table}

Table~\ref{tab:ablation} varies the deployment budget and fixed pin level.
Increasing $\epsmax$ yields monotonically higher gain while each \csrnr{}
certificate remains within its threshold. A fixed $p$ can reach a useful
operating point, but its exploitability is measured after selection rather
than enforced against a target. Here $p=0.7$ exceeds the $0.15$ target by
$65\%$, whereas $p=0.5$ remains below it. A cross-suite holdout selects
$p=0.5$ in both directions; frozen-target gain and certificate are
$+0.370$ and $0.143$ on concentrated opponents, and $+0.017$ and $0.096$ on
diffuse opponents, where 2 of 12 are exploited
(Appendix~\ref{app:repro}).

\subsection{Cross-Game Validation in Solver-Tractable Games}
\label{sec:generality}

We repeat the six-arm comparison and adversarial audit on Liar's Dice and
5-rank Leduc. Liar's Dice uses OpenSpiel with one 3-sided die per player and
pools indexed by own die and current bid. The 5-rank Leduc game has 780
information sets and uses the same pool taxonomy as the main game. Each study
uses 12 opponents. The budgets are $\epsmax=0.05$ for Liar's Dice and $0.15$
for 5-rank Leduc; other scheduling parameters are unchanged.

\begin{table}[t]
\centering
\small
\begin{tabular}{lrrrr}
\toprule
Game & \csrnr{} & Gate & F-Mix cert. & Audit\\
\midrule
Leduc & $+0.209$ & $+0.034$ & $2.04$ ($14\times$) & $0/12$k\\
Liar's Dice & $+0.181$ & $+0.031$ & $0.22$ ($4.4\times$) & $0/12$k\\
5-rank Leduc & $+0.008$ & $+0.002$ & $1.52$ ($10\times$) & $0/12$k\\
\bottomrule
\end{tabular}
\caption{Cross-game summary at $T=800$. Gain columns use 10 seeds per
opponent; audits use three modes, 5 seeds, and 800 hands. Fixed-Mix reports
its largest certificate and budget multiple. The largest \csrnr{}
certificates are $0.150$, $0.003$, and $0.088$, within budgets $0.15$,
$0.05$, and $0.15$, respectively.}
\label{tab:games}
\end{table}

Table~\ref{tab:games} shows that certificate control transfers while gain
follows confirmability. On Liar's Dice, \csrnr{} earns $5.8\times$ the gate
($p=0.004$) and captures from $67\%$ to $100\%$ of oracle gain on confirmed
deviators; its step-like frontier stays inexpensive through approximately
$p=0.7$, whereas high-pin responses to misestimated models reach certificate
$0.44$. Revealed dice remove hidden-card attribution from
Assumption~\ref{ass:estimator}.

On 5-rank Leduc the whole suite is diffuse, so gain follows detection. Under the
unified detector the suite confirms 0, 1, and 4 of 12 opponents at $T=800$,
$2400$, and $5000$; captured safe gain rises with the horizon (mean cumulative
$0.000$, $0.002$, and $0.007$ chips per hand) while the worst deployed
certificate stays within budget ($0.0008$, $0.073$, and $0.138\le0.15$). In
matched stress tests, the unbudgeted control loses $-0.236$ and $-0.714$ per
hand at certificates $1.175$ and $2.064$; the corresponding certified losses
are $-0.0026$ and $-0.0066$ at certificates $0.010$ and $0.060$.

\paragraph{Pre-registered concentrated 5-rank.}
To ask whether the weak 5-rank signal comes from the larger game or from the
diffuse deviations, we pre-register twelve concentrated 5-rank opponents, each
exact Nash with isolated pure-action blunders in a facing-a-bet cell, screened
by a frozen Phase-A gate committed by SHA-256 before any six-arm run. Three
pass and nine are recorded as negatives. On the admitted three an exact best
response captures $1.30$, $1.30$, and $1.39$ chips per hand, yet \csrnr{}
captures $0$: the rarely reached board-pair cell never accumulates enough
observations to confirm at $T=800$, so hero holds Nash at a solver-floor
certificate ($8\times10^{-4}$). Fixed-DBR and Fixed-Mix also capture $0$ and
the binary gate releases once. Confirmation, not certification, is the
bottleneck here, and it recedes with the horizon.

\paragraph{Independently trained opponents.}
We also evaluate twelve Liar's Dice opponents that are not hand-designed:
external-sampling MCCFR checkpoints at four training budgets, three
independent runs each, with twelve distinct strategy hashes and a
pre-registered budget and seed schedule, so no opponent is selected by outcome.
They are near-equilibrium (oracle gain falls from $0.48$ to $0.03$ chips per
hand as the budget grows), so absolute gains are small, and the ordering is
consistent: \csrnr{} captures the most beyond-Nash value (mean final gain
$+0.0065$, positive on $11/12$ and never below the reference), ahead of
Fixed-Mix ($+0.0057$, $8/12$), Fixed-DBR ($+0.0055$, $9/12$), and the binary
gate ($+0.0014$, which never completes a release). Every deployed strategy
stays at the solver-floor certificate ($3\times10^{-4}\ll\epsmax$).

\section{Discussion and Conclusion}
\label{sec:discussion}

\paragraph{Safety is a property of the played strategy.}
The certificate $B_t$ measures the full-tree best-response value of the strategy
the agent actually plays, so model quality moves the captured gain and not the
deployment bound. This is what separates \csrnr{} from methods that mix in a
model-based response by trajectory weight. Fixed-Mix shares the estimator and
matches \csrnr{}'s gain on the main suite, yet a mixture that plays
$\mathrm{BR}(\sigmahat)$ part of the time inherits that response's extreme
exploitability: $2.04$ on the main suite, and $14\times$, $4.4\times$, and
$10\times$ its budget in Leduc, Liar's Dice, and 5-rank Leduc, while every
\csrnr{} deployment stays inside budgets $0.15$, $0.05$, and $0.15$.

\paragraph{Certification is cheap enough to run online.}
A candidate level costs one restricted solve and one exact best-response pass.
Single-threaded on a Xeon 6348, the certificate pass takes $3.07$\,ms in Leduc,
$12.27$\,ms in 5-rank Leduc, and $22.25$\,ms in Liar's Dice against solves of
$6.62$\,s, $34.81$\,s, and $3.87$\,s, two to three orders of magnitude less than
the response it guards. The schedule can therefore re-certify every candidate at
every checkpoint.

\paragraph{Conclusion.}
\csrnr{} is the first method to make reference-relative safety a computed
property of the deployed strategy rather than an assumption about the model
that proposed it, set by a single budget $\epsmax$ spanning equilibrium play
and an unrestricted schedule. All 36{,}000 audited hands stayed within budget with no
violation, and where the data confirmed deviations \csrnr{} earned $6.2\times$
the binary gate's gain in Leduc and $5.8\times$ in Liar's Dice. Safe deployment
is governed by the strategy that is played rather than confidence in the model
that proposed it.

\bibliographystyle{plainnat}
\bibliography{references}

\begin{thebibliography}{54}
\providecommand{\natexlab}[1]{#1}
\providecommand{\url}[1]{\texttt{#1}}
\expandafter\ifx\csname urlstyle\endcsname\relax
  \providecommand{\doi}[1]{doi: #1}\else
  \providecommand{\doi}{doi: \begingroup \urlstyle{rm}\Url}\fi

\bibitem[Bernasconi-de Luca et~al.(2021)Bernasconi-de Luca, Cacciamani,
  Fioravanti, Gatti, Marchesi, and Trov{\`o}]{bernasconi2021exploiting}
Martino Bernasconi-de Luca, Federico Cacciamani, Simone Fioravanti, Nicola
  Gatti, Alberto Marchesi, and Francesco Trov{\`o}.
\newblock Exploiting opponents under utility constraints in sequential games.
\newblock \emph{Advances in Neural Information Processing Systems},
  34:\penalty0 13177--13188, 2021.

\bibitem[Bowling et~al.(2015)Bowling, Burch, Johanson, and
  Tammelin]{bowling2015heads}
Michael Bowling, Neil Burch, Michael Johanson, and Oskari Tammelin.
\newblock Heads-up limit hold'em poker is solved.
\newblock \emph{Science}, 347\penalty0 (6218):\penalty0 145--149, 2015.

\bibitem[Brown and Sandholm(2015)]{brown2015regret}
Noam Brown and Tuomas Sandholm.
\newblock Regret-based pruning in extensive-form games.
\newblock \emph{Advances in neural information processing systems}, 28, 2015.

\bibitem[Brown and Sandholm(2017)]{brown2017safe}
Noam Brown and Tuomas Sandholm.
\newblock Safe and nested subgame solving for imperfect-information games.
\newblock \emph{Advances in neural information processing systems}, 30, 2017.

\bibitem[Brown and Sandholm(2018)]{brown2018superhuman}
Noam Brown and Tuomas Sandholm.
\newblock Superhuman ai for heads-up no-limit poker: Libratus beats top
  professionals.
\newblock \emph{Science}, 359\penalty0 (6374):\penalty0 418--424, 2018.

\bibitem[Brown and Sandholm(2019{\natexlab{a}})]{brown2019solving}
Noam Brown and Tuomas Sandholm.
\newblock Solving imperfect-information games via discounted regret
  minimization.
\newblock In \emph{Proceedings of the AAAI Conference on Artificial
  Intelligence}, volume~33, pages 1829--1836, 2019{\natexlab{a}}.

\bibitem[Brown and Sandholm(2019{\natexlab{b}})]{brown2019superhuman}
Noam Brown and Tuomas Sandholm.
\newblock Superhuman ai for multiplayer poker.
\newblock \emph{Science}, 365\penalty0 (6456):\penalty0 885--890,
  2019{\natexlab{b}}.

\bibitem[Brown et~al.(2017)Brown, Kroer, and Sandholm]{brown2017dynamic}
Noam Brown, Christian Kroer, and Tuomas Sandholm.
\newblock Dynamic thresholding and pruning for regret minimization.
\newblock In \emph{Proceedings of the AAAI conference on artificial
  intelligence}, volume~31, 2017.

\bibitem[Brown et~al.(2018)Brown, Sandholm, and Amos]{brown2018depth}
Noam Brown, Tuomas Sandholm, and Brandon Amos.
\newblock Depth-limited solving for imperfect-information games.
\newblock \emph{Advances in neural information processing systems}, 31, 2018.

\bibitem[Brown et~al.(2019)Brown, Lerer, Gross, and Sandholm]{brown2019deep}
Noam Brown, Adam Lerer, Sam Gross, and Tuomas Sandholm.
\newblock Deep counterfactual regret minimization.
\newblock In \emph{International conference on machine learning}, pages
  793--802. PMLR, 2019.

\bibitem[Brown et~al.(2020)Brown, Bakhtin, Lerer, and Gong]{brown2020combining}
Noam Brown, Anton Bakhtin, Adam Lerer, and Qucheng Gong.
\newblock Combining deep reinforcement learning and search for
  imperfect-information games.
\newblock \emph{Advances in neural information processing systems},
  33:\penalty0 17057--17069, 2020.

\bibitem[Caen et~al.(2026)Caen, Winands, and Soemers]{caen2026stratformer}
Andy Caen, Mark~HM Winands, and Dennis~JNJ Soemers.
\newblock Stratformer: Adaptive opponent modeling and exploitation in
  imperfect-information games.
\newblock \emph{arXiv preprint arXiv:2604.25796}, 2026.

\bibitem[Capitaine et~al.(2026)Capitaine, Scheid, Boursier, Durmus, and
  Jordan]{capitaine2026test}
Aymeric Capitaine, Antoine Scheid, Etienne Boursier, Alain Durmus, and
  Michael~I Jordan.
\newblock Test-then-punish: A statistical approach to repeated games.
\newblock \emph{arXiv preprint arXiv:2603.05619}, 2026.

\bibitem[Farina et~al.(2021)Farina, Kroer, and Sandholm]{farina2021faster}
Gabriele Farina, Christian Kroer, and Tuomas Sandholm.
\newblock Faster game solving via predictive blackwell approachability:
  Connecting regret matching and mirror descent.
\newblock In \emph{Proceedings of the AAAI Conference on Artificial
  Intelligence}, volume~35, pages 5363--5371, 2021.

\bibitem[Farina et~al.(2023)Farina, Grand-Cl{\'e}ment, Kroer, Lee, and
  Luo]{farina2023regret}
Gabriele Farina, Julien Grand-Cl{\'e}ment, Christian Kroer, Chung-Wei Lee, and
  Haipeng Luo.
\newblock Regret matching+:(in) stability and fast convergence in games.
\newblock \emph{Advances in Neural Information Processing Systems},
  36:\penalty0 61546--61572, 2023.

\bibitem[Fu et~al.(2022)Fu, Tian, Yu, Liu, Wu, Xiong, Wen, Li, Xing, Fu,
  et~al.]{fu2022greedy}
Haobo Fu, Ye~Tian, Hongxiang Yu, Weiming Liu, Shuang Wu, Jiechao Xiong, Ying
  Wen, Kai Li, Junliang Xing, Qiang Fu, et~al.
\newblock Greedy when sure and conservative when uncertain about the opponents.
\newblock In \emph{International Conference on Machine Learning}, pages
  6829--6848. PMLR, 2022.

\bibitem[Ganzfried(2025{\natexlab{a}})]{DBLP:journals/corr/abs-2508-17671}
Sam Ganzfried.
\newblock Consistent opponent modeling of static opponents in
  imperfect-information games.
\newblock \emph{CoRR}, abs/2508.17671, 2025{\natexlab{a}}.

\bibitem[Ganzfried(2025{\natexlab{b}})]{ganzfried2025nonparametric}
Sam Ganzfried.
\newblock Nonparametric strategy test.
\newblock In \emph{The International FLAIRS Conference Proceedings}, volume~38,
  2025{\natexlab{b}}.

\bibitem[Ganzfried and Sandholm(2015)]{ganzfried2015safe}
Sam Ganzfried and Tuomas Sandholm.
\newblock Safe opponent exploitation.
\newblock \emph{ACM Transactions on Economics and Computation (TEAC)},
  3\penalty0 (2):\penalty0 1--28, 2015.

\bibitem[Gauthier et~al.(2026)Gauthier, Bach, and Jordan]{gauthier2026betting}
Etienne Gauthier, Francis Bach, and Michael~I Jordan.
\newblock Betting on equilibrium: Monitoring strategic behavior in multi-agent
  systems.
\newblock \emph{arXiv preprint arXiv:2601.05427}, 2026.

\bibitem[Ge et~al.(2022)Ge, Yang, Tian, Chen, and Gao]{ge2022modeling}
Zhenxing Ge, Shangdong Yang, Pinzhuo Tian, Zixuan Chen, and Yang Gao.
\newblock Modeling rationality: Toward better performance against unknown
  agents in sequential games.
\newblock \emph{IEEE Transactions on Cybernetics}, 54\penalty0 (5):\penalty0
  2966--2977, 2022.

\bibitem[Ge et~al.(2023)Ge, Xu, Ding, Li, and Gao]{ge2023efficient}
Zhenxing Ge, Zheng Xu, Tianyu Ding, Wenbin Li, and Yang Gao.
\newblock Efficient subgame refinement for extensive-form games.
\newblock \emph{Advances in Neural Information Processing Systems},
  36:\penalty0 8280--8291, 2023.

\bibitem[Ge et~al.(2024)Ge, Xu, Ding, Meng, An, Li, and Gao]{ge2024safe}
Zhenxing Ge, Zheng Xu, Tianyu Ding, Linjian Meng, Bo~An, Wenbin Li, and Yang
  Gao.
\newblock Safe and robust subgame exploitation in imperfect information games.
\newblock In \emph{Forty-first International Conference on Machine Learning},
  2024.

\bibitem[Gr{\"u}nwald(2024)]{grunwald2024proposer}
Peter Gr{\"u}nwald.
\newblock Proposer of the vote of thanks to waudy-smith and ramdas and
  contribution to the discussion of `estimating means of bounded random
  variables by betting'.
\newblock \emph{Journal of the Royal Statistical Society Series B: Statistical
  Methodology}, 86\penalty0 (1):\penalty0 28--30, 2024.

\bibitem[Hoehn et~al.(2005)Hoehn, Southey, Holte, and
  Bulitko]{hoehn2005effective}
Bret Hoehn, Finnegan Southey, Robert~C Holte, and Valeriy Bulitko.
\newblock Effective short-term opponent exploitation in simplified poker.
\newblock In \emph{AAAI}, volume~5, pages 783--788, 2005.

\bibitem[Howard et~al.(2021)Howard, Ramdas, McAuliffe, and
  Sekhon]{howard2021time}
Steven~R Howard, Aaditya Ramdas, Jon McAuliffe, and Jasjeet Sekhon.
\newblock Time-uniform, nonparametric, nonasymptotic confidence sequences.
\newblock \emph{The Annals of Statistics}, 49\penalty0 (2):\penalty0
  1055--1080, 2021.

\bibitem[Jing et~al.(2023)Jing, Li, Liu, Zang, Fu, Fu, Xing, and
  Cheng]{jing2023towards}
Yuheng Jing, Kai Li, Bingyun Liu, Yifan Zang, Haobo Fu, Qiang Fu, Junliang
  Xing, and Jian Cheng.
\newblock Towards offline opponent modeling with in-context learning.
\newblock In \emph{The Twelfth International Conference on Learning
  Representations}, 2023.

\bibitem[Jing et~al.(2024)Jing, Liu, Li, Zang, Fu, Fu, Xing, and
  Cheng]{jing2024opponent}
Yuheng Jing, Bingyun Liu, Kai Li, Yifan Zang, Haobo Fu, Qiang Fu, Junliang
  Xing, and Jian Cheng.
\newblock Opponent modeling with in-context search.
\newblock \emph{Advances in Neural Information Processing Systems},
  37:\penalty0 61549--61591, 2024.

\bibitem[Jing et~al.(2025)Jing, Li, Liu, Fu, Fu, Xing, and Cheng]{jing2025open}
Yuheng Jing, Kai Li, Bingyun Liu, Haobo Fu, Qiang Fu, Junliang Xing, and Jian
  Cheng.
\newblock An open-ended learning framework for opponent modeling.
\newblock In \emph{Proceedings of the AAAI Conference on Artificial
  Intelligence}, volume~39, pages 23222--23230, 2025.

\bibitem[Johanson and Bowling(2009)]{johanson2009data}
Michael Johanson and Michael Bowling.
\newblock Data biased robust counter strategies.
\newblock In \emph{Artificial Intelligence and Statistics}, pages 264--271.
  PMLR, 2009.

\bibitem[Johanson et~al.(2007)Johanson, Zinkevich, and
  Bowling]{johanson2007computing}
Michael Johanson, Martin Zinkevich, and Michael~H. Bowling.
\newblock Computing robust counter-strategies.
\newblock In \emph{{NIPS}}, pages 721--728. Curran Associates, Inc., 2007.

\bibitem[Li and Huang(2025)]{li2025efficient}
Boning Li and Longbo Huang.
\newblock Efficient online pruning and abstraction for imperfect information
  extensive-form games.
\newblock In \emph{The Thirteenth International Conference on Learning
  Representations}, 2025.

\bibitem[Li and Huang(2026{\natexlab{a}})]{li2026effective}
Boning Li and Longbo Huang.
\newblock Effective, efficient, and general information abstraction for
  imperfect-information extensive-form games.
\newblock \emph{arXiv preprint arXiv:2605.10900}, 2026{\natexlab{a}}.

\bibitem[Li and Huang(2026{\natexlab{b}})]{li2026real}
Boning Li and Longbo Huang.
\newblock Real-time parallel counterfactual regret minimization.
\newblock \emph{arXiv preprint arXiv:2605.19928}, 2026{\natexlab{b}}.

\bibitem[Li et~al.(2024)Li, Fang, and Huang]{li2024rl}
Boning Li, Zhixuan Fang, and Longbo Huang.
\newblock Rl-cfr: improving action abstraction for imperfect information
  extensive-form games with reinforcement learning.
\newblock In \emph{Proceedings of the 41st International Conference on Machine
  Learning}, pages 27752--27770, 2024.

\bibitem[Li et~al.(2026{\natexlab{a}})Li, Chen, and
  Huang]{li2026correlatedchancesamplingmonte}
Boning Li, Yu~Chen, and Longbo Huang.
\newblock Correlated chance sampling for monte carlo counterfactual regret
  minimization, 2026{\natexlab{a}}.
\newblock URL \url{https://arxiv.org/abs/2607.27035}.

\bibitem[Li et~al.(2026{\natexlab{b}})Li, Wang, and Huang]{li2026pokerskill}
Boning Li, Baoxiang Wang, and Longbo Huang.
\newblock Pokerskill: Llms can play expert-level poker without training or
  solvers.
\newblock \emph{arXiv preprint arXiv:2605.30094}, 2026{\natexlab{b}}.

\bibitem[Liu et~al.(2022)Liu, Wu, Liu, Jing, Yang, Tang, and
  Zhang]{liu2022safe}
Mingyang Liu, Chengjie Wu, Qihan Liu, Yansen Jing, Jun Yang, Pingzhong Tang,
  and Chongjie Zhang.
\newblock Safe opponent-exploitation subgame refinement.
\newblock \emph{Advances in Neural Information Processing Systems},
  35:\penalty0 27610--27622, 2022.

\bibitem[Liu et~al.(2023)Liu, Fu, Fu, and Wei]{liu2023opponent}
Weiming Liu, Haobo Fu, Qiang Fu, and Yang Wei.
\newblock Opponent-limited online search for imperfect information games.
\newblock In \emph{International Conference on Machine Learning}, pages
  21567--21585. PMLR, 2023.

\bibitem[McCracken and Bowling(2004)]{mccracken2004safe}
Peter McCracken and Michael Bowling.
\newblock Safe strategies for agent modelling in games.
\newblock In \emph{AAAI Technical Report (2)}, pages 103--110, 2004.

\bibitem[Meng et~al.(2025)Meng, Yang, Zhang, Ge, Yang, Ding, Li, An, and
  Gao]{meng2025efficient}
Linjian Meng, Tianpei Yang, Youzhi Zhang, Zhenxing Ge, Shangdong Yang, Tianyu
  Ding, Wenbin Li, Bo~An, and Yang Gao.
\newblock Efficient last-iterate convergence in solving extensive-form games.
\newblock In \emph{The Thirty-ninth Annual Conference on Neural Information
  Processing Systems}, 2025.
\newblock URL \url{https://openreview.net/forum?id=JzWtqd9CGJ}.

\bibitem[Meng et~al.(2026)Meng, Yang, Zhang, Ge, and Gao]{meng2026faster}
Linjian Meng, Tianpei Yang, Youzhi Zhang, Zhenxing Ge, and Yang Gao.
\newblock Faster game solving via asymmetry of step sizes.
\newblock In \emph{Proceedings of the AAAI Conference on Artificial
  Intelligence}, volume~40, pages 17161--17169, 2026.

\bibitem[Milec et~al.(2025)Milec, Kova{\v{r}}{\'\i}k, and
  Lis{\`y}]{milec2025adapting}
David Milec, Vojt{\v{e}}ch Kova{\v{r}}{\'\i}k, and Viliam Lis{\`y}.
\newblock Adapting beyond the depth limit: Counter strategies in large
  imperfect information games.
\newblock In \emph{Proceedings of the 24th International Conference on
  Autonomous Agents and Multiagent Systems}, pages 2675--2677, 2025.

\bibitem[Morav{\v{c}}{\'\i}k et~al.(2017)Morav{\v{c}}{\'\i}k, Schmid, Burch,
  Lis{\`y}, Morrill, Bard, Davis, Waugh, Johanson, and
  Bowling]{moravvcik2017deepstack}
Matej Morav{\v{c}}{\'\i}k, Martin Schmid, Neil Burch, Viliam Lis{\`y}, Dustin
  Morrill, Nolan Bard, Trevor Davis, Kevin Waugh, Michael Johanson, and Michael
  Bowling.
\newblock Deepstack: Expert-level artificial intelligence in heads-up no-limit
  poker.
\newblock \emph{Science}, 356\penalty0 (6337):\penalty0 508--513, 2017.

\bibitem[Murgoci et~al.(2026)Murgoci, Spaan, and
  Oren]{murgoci2026alphaexploitem}
Vlad Murgoci, Matthijs Spaan, and Yaniv Oren.
\newblock Alphaexploitem: Going beyond the nash equilibrium in poker by
  learning to exploit suboptimal play.
\newblock \emph{arXiv preprint arXiv:2605.09150}, 2026.

\bibitem[Nashed and Zilberstein(2022)]{nashed2022survey}
Samer Nashed and Shlomo Zilberstein.
\newblock A survey of opponent modeling in adversarial domains.
\newblock \emph{Journal of Artificial Intelligence Research}, 73:\penalty0
  277--327, 2022.

\bibitem[Ponsen et~al.(2010)Ponsen, Lanctot, and De~Jong]{ponsen2010mcrnr}
Marc Ponsen, Marc Lanctot, and Steven De~Jong.
\newblock Mcrnr: fast computing of restricted nash responses by means of
  sampling.
\newblock In \emph{Proceedings of the 3rd AAAI Conference on Interactive
  Decision Theory and Game Theory}, pages 43--49, 2010.

\bibitem[Southey et~al.(2005)Southey, Bowling, Larson, Piccione, Burch,
  Billings, and Rayner]{southey2005bayes}
Finnegan Southey, Michael Bowling, Bryce Larson, Carmelo Piccione, Neil Burch,
  Darse Billings, and Chris Rayner.
\newblock Bayes' bluff: opponent modelling in poker.
\newblock In \emph{Proceedings of the Twenty-First Conference on Uncertainty in
  Artificial Intelligence}, pages 550--558, 2005.

\bibitem[Sychrovsk{\`y} et~al.(2024)Sychrovsk{\`y}, {\v{S}}ustr, Davoodi,
  Bowling, Lanctot, and Schmid]{sychrovsky2024learning}
David Sychrovsk{\`y}, Michal {\v{S}}ustr, Elnaz Davoodi, Michael Bowling, Marc
  Lanctot, and Martin Schmid.
\newblock Learning not to regret.
\newblock In \emph{Proceedings of the AAAI Conference on Artificial
  Intelligence}, volume~38, pages 15202--15210, 2024.

\bibitem[Tammelin(2014)]{tammelin2014solving}
Oskari Tammelin.
\newblock Solving large imperfect information games using cfr+.
\newblock \emph{arXiv preprint arXiv:1407.5042}, 2014.

\bibitem[Wang et~al.(2025)Wang, Wang, and Song]{wang2025horse}
Shijia Wang, Jiao Wang, and Bangyan Song.
\newblock Horse-cfr: Hierarchical opponent reasoning for safe exploitation
  counterfactual regret minimization.
\newblock \emph{Expert Systems with Applications}, 263:\penalty0 125697, 2025.

\bibitem[Xu et~al.(2024)Xu, Li, Fu, Fu, Xing, and Cheng]{xu2024dynamic}
Hang Xu, Kai Li, Haobo Fu, Qiang Fu, Junliang Xing, and Jian Cheng.
\newblock Dynamic discounted counterfactual regret minimization.
\newblock In \emph{The Twelfth International Conference on Learning
  Representations}, 2024.

\bibitem[Xu et~al.(2026)Xu, Li, Fu, Fu, Xing, and Cheng]{xu2026deep}
Hang Xu, Kai Li, Haobo Fu, Qiang Fu, Junliang Xing, and Jian Cheng.
\newblock Deep (predictive) discounted counterfactual regret minimization.
\newblock In \emph{Proceedings of the AAAI Conference on Artificial
  Intelligence}, volume~40, pages 17284--17292, 2026.

\bibitem[Zinkevich et~al.(2007)Zinkevich, Johanson, Bowling, and
  Piccione]{zinkevich2007regret}
Martin Zinkevich, Michael Johanson, Michael Bowling, and Carmelo Piccione.
\newblock Regret minimization in games with incomplete information.
\newblock \emph{Advances in neural information processing systems}, 20, 2007.

\end{thebibliography}

\appendix

\section{Reproducibility Details}
\label{app:repro}

\paragraph{Game.}
Leduc hold'em \citep{southey2005bayes}: 6 cards (3 ranks, 2 suits), unit
ante, two betting rounds with raise sizes 2 and 4, at most two raises
per round; 288 information sets. The equilibrium-reference strategy is
the average CFR$^+$ strategy after 4000 iterations. Expected values and
best-response values of fixed behavioural strategies are then evaluated
by full tree walks, without Monte Carlo evaluation noise; the
finite-iteration reference is not claimed to be an exact Nash
equilibrium.

\paragraph{Pools ($m=30$ scanned cells).}
Opponent observations aggregate into 12 pools: preflop by private card
$\times$ \{not facing / facing a raise\} (6 pools), postflop by hand
strength category \{pair, high, low\} $\times$ \{not facing / facing\}
(6 pools). Not-facing pools have canonical actions (check, raise);
facing pools (fold, call, raise), giving $6{\times}2 + 6{\times}3 = 30$
cells. Hidden cards are soft-attributed over board-consistent holdings;
the equilibrium reference $f^*$ per pool is computed by passing a fixed
20{,}000-hand Nash-vs-Nash sample (seed 0) through the same estimator.

\paragraph{Detection.}
Episode-wide error budget $\alpha=0.05$, Bonferroni-split over the
$m=30$ cells. Main comparison: stitched sub-Gaussian boundary
(variance proxy $1/4$), margin $\delta=0.10$ (matching the binary
gate's margin). Long-horizon runs: empirical-Bernstein boundary with
$\delta=0.05$. Both use the stitching constants of
\citet{howard2021time}:
$u(v) = 1.7\sqrt{v(\ln\ln 2v + 0.72\ln(5.2/\alpha_{\mathrm{cell}}))}$,
plus $3.4(\ln\ln 2v + 0.72 \ln (5.2/\alpha_{\mathrm{cell}}))$ for the
empirical-Bernstein form.

\paragraph{Schedule.}
Pin grid $p \in \{0,\allowbreak 0.1,\allowbreak 0.2,\allowbreak
0.3,\allowbreak 0.5,\allowbreak 0.7,\allowbreak 0.9\}$;
checkpoints at hands $\{40,\allowbreak 80,\allowbreak 160,\allowbreak
320,\allowbreak 480,\allowbreak 640\}$ for $T{=}800$
and $\{160,\allowbreak 320,\allowbreak 640,\allowbreak 1280,\allowbreak
2560,\allowbreak 4000\}$ for $T{=}5000$; at most two grid
steps upward per checkpoint; restricted solves run CFR$^+$ for 400
iterations; each candidate level is certified by one exact
best-response pass. Budget $\epsmax = 0.15$ unless stated.

\paragraph{Binary gate (baseline).}
GATE-style money-verified release: frequency test (Wilson intervals with
anytime peeling, $\alpha=0.05$, margin $0.10$) proposes candidate pools;
a candidate stable for 6 consecutive refreshes is probed with
probability $0.1$ per hand by playing
$\mathrm{BR}(\sigmahat_{\mathrm{cand}})$; release requires the realized
probe-payoff lower confidence bound to clear the Nash floor. Identical
configuration at every horizon.

\paragraph{Evaluation and audit.}
All gains are exact expected values of the played profile (no
Monte-Carlo evaluation noise); observations for learning are sampled
hands. Seeds: 10 per opponent (main, $T{=}2400$, ablations), 6
($T{=}5000$), 5 (stress modes). The per-hand certificate audit checks
$\mathrm{gain}_t \ge -\Bt - 5\times10^{-3}$, the tolerance covering
CFR$^+$ solve residue at 400 iterations; ``0 violations'' is reported
under this tolerance; the stress-test boundary of Figure~\ref{fig:stress}
matches to ${\sim}10^{-4}$. Everything runs on CPU (Intel Xeon Gold 6348 at
2.60\,GHz, 251\,GB RAM, Ubuntu 20.04) under Python 3.9.16 with NumPy 1.26.4 as
the only dependency; the main comparison costs about 8 core-hours, close to one
hour of wall time with the 24 opponents in parallel.

\paragraph{Price of confirmation (horizon detail).}
Under the empirical-Bernstein detector of the horizon sweep shown in the main
paper, \csrnr{} releases 1, 5, and 7
of 12 diffuse opponents at $T\in\{800,2400,5000\}$; the main comparison table
instead uses a sub-Gaussian boundary and releases 2 at $T{=}800$
(empirical-Bernstein is looser at small samples and tighter later). The 800
hands over 12 pools give from 40 to 155 observations per pool. At $T{=}2400$
the strong over-caller yields $+0.163$/hand for \csrnr{} at certificate
$0.148$ while the gate loses $-0.061$/hand; at $T{=}5000$ the gate's higher
mean comes from one unbounded release gaining $+0.65$/hand versus the budgeted
$+0.20$/hand. The gate has from 1 to 3 losing opponents over these horizons;
\csrnr{} has none at $T{=}800$ or $2400$, and its two marginal releases at
$T{=}5000$ have cumulative means $-0.0022$ and $-0.0018$/hand while every
deployed strategy remains within budget.

\paragraph{Liar's Dice (generality study).}
OpenSpiel \texttt{liars\_dice}, one die per player, \texttt{dice\_sides}
$=3$: 192 information sets, 7 actions (6 bids in count-major order plus
``Liar''), payoffs $\pm 1$. The highest face is wild (counts toward
every face); verified against all 54 terminal bid and challenge pairs. Pools:
own die $\times$ current highest bid (18 pools); cells are the concrete
continuation actions (bids above the current one, plus Liar; 63 cells,
Bonferroni across all). Every hand ends with a challenge and both dice
revealed, so opponent counts are exact identity counts (no soft
attribution), and the Nash reference is the exact reach-weighted tree
walk. Restricted solves run the same CFR$^+$ (400 iterations) on a
one-shot materialization of the pyspiel tree; certificates use
OpenSpiel's exact best response. Budget $\epsmax{=}0.05$; grid,
checkpoints, $\alpha$, and margin as in the main configuration.
Opponent suite (12; exact best-response value above the Nash floor in
parentheses): class-biased log-odds liar-caller and bluffer at two
strengths (from $0.03$ to $0.13$); softened-challenge ``trusting'' at
$q{=}0.3/0.6$ ($0.20/0.40$); truth-bid-biased ``honest'' at $q{=}0.4/0.7$
($0.10/0.19$); one dominated-action leak blended at $w{=}0.35/0.7/1.0$
and a two-cell variant (from $0.16$ to $0.67$). The binary-gate arm is the
same money-verified gate at per-infoset granularity (its game-agnostic
form; payoff range $R{=}2.0$), probing and releasing through the
identical online loop. The decoy for the adaptive stress modes is
trusting-$0.6$ (the diffuse log-odds deviators sit below the $T{=}800$
detection radar and would never bait a release).

\paragraph{5-rank Leduc (generality study).}
$n_{\mathrm{ranks}}{=}5$, two suits: $780$ information sets; identical
pool taxonomy ($16$ pools, $40$ cells via the same private-card
$\times$ facing construction), identical budget, grid, and detector
configuration as 3-rank. The 12-opponent suite is the log-odds
construction instantiated on the 5-rank equilibrium. The matched
unbudgeted control uses the same detector with $\epsmax{=}\infty$ and
$p\in\{0,1\}$; zero audit violations for this arm are vacuous because
its certificates reach $1.175$ and $2.064$, rather than satisfying the
nominal $0.15$ budget. One latent-deck
pitfall matters for reproduction: the Leduc trajectory sampler and the
gate's reach-weighted reference construct a default 3-rank deck unless
the game config is passed explicitly; all 5-rank runs here pass it (the
drivers and a config-aware reference), which we verified by deck
enumeration (unpatched, only ranks 0, 1, and 2 are ever dealt and the
reference diverges by up to $0.50$ in rank-4 pools).

\paragraph{Fixed-pin holdout and runtime.}
The fixed-DBR sweep uses
$p\in\{0,\allowbreak 0.05,\allowbreak 0.1,\allowbreak 0.2,\allowbreak
0.3,\allowbreak 0.5,\allowbreak 0.7,\allowbreak 0.9,\allowbreak 1\}$, 10 seeds,
and opponent-equal
weighting. We select the highest-gain source-suite pin among those whose
source worst certificate is at most $0.15$, then evaluate that frozen pin
on the other suite. Both directions select $p{=}0.5$; the diffuse and
concentrated suite gain/certificate pairs are respectively
$0.0169/0.0955$ and $0.3700/0.1434$. This is a grid- and suite-level
holdout, not a population guarantee. The solve and certificate timings quoted
in Section~\ref{sec:discussion} are medians over five sequential warm-context
repeats after
one 10-iteration warmup, measured single-threaded on a Xeon 6348; we do not
compare speed across the three game-specific implementations.

\paragraph{Independently-trained opponents.}
Beyond the scripted suites, we evaluate on $12$ Liar's Dice opponents
that are \emph{not} hand-designed: fresh external-sampling MCCFR
checkpoints at four training budgets ($32,128,512,2048$ iterations,
three independent runs each; all $12$ strategy hashes distinct). The
campaign is pre-registered. Budgets, seeds, and the run schedule are fixed
before evaluation and verified by hash, so no opponent is selected by outcome. Each checkpoint is used as the fixed P1 strategy and scored
by exact tree walks ($10$ seeds, $T{=}800$, $\epsmax{=}0.05$). These
opponents are near-equilibrium (oracle gain falls from $0.48$ to $0.03$
chips/hand as the budget grows), so absolute gains are small, but the
ordering is consistent: \csrnr{} captures the most beyond-Nash value
(mean final gain $+0.0065$ over the equilibrium reference, positive on
$11/12$ and never below it), ahead of Fixed-DBR ($+0.0055$, $9/12$),
Fixed-Mix ($+0.0057$, $8/12$), and the binary gate ($+0.0014$, which
never completes a release). Every deployed strategy stays at the
solver-floor certificate ($3\times10^{-4}\ll\epsmax$). This is a
confirmatory robustness check on unscripted opponents, not a
large-gain regime.

\paragraph{Concentrated 5-rank opponents (pre-registered).}
To separate ``the game is larger'' from ``the deviations are diffuse'' as the
cause of the weak 5-rank signal, we pre-registered a concentrated suite: twelve
5-rank opponents, each the exact Nash with one or more isolated pure-action
blunders in a facing-a-bet cell, spanning $\{$preflop, postflop$\}\times\{$over-fold
strong, over-call weak, flat the nuts, single-cell, stacked$\}$. Admission is by
a frozen Phase-A gate (the identical arithmetic and thresholds as the 3-rank
B-path pre-gate: exploitable, $\mathrm{oracle}\ge0.05$; its pool
(near-)zero-data across $N{=}20$ probe records, rate ${\ge}0.5$; and a numeric
blind gap ${\ge}0.05$), committed by SHA-256 before any six-arm run; nine
candidates are recorded as honest negatives, never dropped silently. Only three are admitted. All are board-pair leaks in the only 5-rank pool
rare enough at
$N{=}20$; the six per-card preflop and two overcard leaks are rejected because
their pools \emph{are} reached (zero-data rate from $0.12$ to $0.29$), and flatting the
nuts is sub-threshold ($\mathrm{oracle}=0.046$). On the admitted three the exact
best response captures $1.30/1.30/1.39$ chips/hand, yet \csrnr{} captures $0$: it
never accumulates enough observations of the rarely-reached pair cell to confirm
at $T{=}800$, so it holds Nash at a solver-floor certificate
($8\times10^{-4}\ll0.15$); Fixed-DBR and Fixed-Mix likewise capture $0$, and the
binary gate releases once (on the triple-leak, $+0.059$). Under the matched
empirical-Bernstein detector ($\delta{=}0.05$) the diffuse 5-rank suite is
confirmation-starved at short horizons but recovers with data: \csrnr{} confirms
$0/12$ opponents at $T{=}800$, $1/12$ at $T{=}2400$, and $4/12$ at $T{=}5000$
(the exact best response exploits all twelve, mean $0.38$, at every horizon).
Captured safe gain rises accordingly ($0.000\!\to\!0.002\!\to\!0.007$ chips/hand,
mean cumulative), and every deployed level stays within budget (worst certificate
$0.0008\!\to\!0.073\!\to\!0.138\le0.15$). At $T{=}5000$ \csrnr{} captures the most
budget-respecting beyond-Nash value ($0.007$ versus Fixed-DBR $0.002$); Fixed-Mix's
larger $0.037$ comes at a $13\times$ budget violation (certificate $1.98$), and the
binary gate's five releases dip negative (worst cumulative lower bound $-0.050$).
The 5-rank results consistently show confirmation starvation across leak
geometries. The detection limit recedes with the horizon while certificate
control remains unchanged.

\paragraph{Detection and certification depend on different resources.}
Across the horizon studies one data-limited quantity moves with sample size:
the number of opponents whose deviations separate from the equilibrium
reference. On the diffuse 5-rank suite this rises from $0$ to $1$ to $4$ of
$12$ as $T$ grows to $5000$, and on the main suite the empirical-Bernstein
detector confirms $1$, $5$, and $7$ of $12$ over the same horizons. The
certificate budget is unchanged throughout: every deployed level stays within
$\epsmax$ regardless of how much has been confirmed. Detection governs how much
value is available to capture, the certificate governs how much downside that
capture can cost, and only the former depends on the horizon.

\section{Future Work}
Current exact certificates require a solver-tractable game and a full-tree
best-response pass. A central direction is to extend the certificate to sound
approximate best responses suitable for poker scale
\citep{moravvcik2017deepstack,brown2018superhuman}, where the opponents worth
exploiting now include prompted language-model agents that reach strong play
without training or a solver \citep{li2026pokerskill}. More general pool
construction and additional observation mechanisms may reduce confirmation
starvation, particularly for zero-reach leaks invisible to on-policy data.
Realized-gain rollback may improve profitability after marginal releases, but
would complement rather than replace the played-strategy certificate.

\end{document}